\documentclass[cmp]{svjour}
\usepackage{graphics}
\begin{document}
\title{Connection Between Kinetic Energy And Vorticity Blow-up in 3-D Navier-Stokes Fluid.}
\author{Manuel Garc\'ia-Casado
}                     
\institute{Instituto de Cat\'alisis y Petroleoqu\'imica, CSIC, Marie Curie, 2, Madrid E-28049, Spain.\\
\email{mgarciacasado@icp.csic.es}\\
Telephone: +34 91 585 48 81\\
Fax: +34 91 585 47 79
}

\date{}
%
\communicated{}
\maketitle
\begin{abstract}
In this paper the author formulates and proves a theorem that relates smoothness of kinetic energy
and smoothness of vorticity in a 3-D Navier-Stokes fluid. Setting velocity and vorticity boundary conditions, a direct relation
arises between kinetic energy and the squared Euclidean norm of the vorticity. As a direct result, their smoothness is closely
related. 
\end{abstract}
\section{Introduction}
\label{intro}
It's well known that Navier-Stokes equations give us the behaviour of fluids \cite{RefB1}-\cite{RefB2}. The most general equations of
a viscous compressible fluid is given by the nonlinear partial differential equations
\begin{equation}
  \partial_{t} \vec{u} + \vec{u} \cdot \vec{\nabla} \vec{u} = \nu \Delta 	   \vec{u}-\frac{1}{\rho}\vec{\nabla}p+\left(\frac{\zeta}{\rho}+\frac{1}{3}\nu\right)\vec{\nabla}\left(\vec{\nabla}\cdot\vec{u}\right)+ \frac{1}{\rho}\vec{f},
\label{eq compressible NS}
\end{equation}
where $\rho$ is the density of the fluid, $\nu>0$ is the kinematic viscosity, $\zeta$ is a second viscosity coefficient and $\vec{f}$ is an external force. In equation (\ref{eq compressible NS}), the velocity $\vec{u}\left(\vec{x},t\right)$ of a fluid particle is a vectorial function $\vec{u}:\Omega\times\left[0,\infty\right)\longmapsto\bf {R}^{3}$, and pressure $p\left(\vec{x},t\right)$ is a scalar function $p:\Omega\times\left[0,\infty\right)\longmapsto\bf{R}$, where $\vec{x}\in\Omega\subset\bf {R}^{3}$ and $t\in\left[0,\infty \right)$. The fluid is 
incompressible if $\rho$ is constant and, as a mass conservation law consequence, the flow is free-divergence, i.e. $\vec{\nabla}\cdot\vec{u}=0$. Then, using Einstein notation, the Navier-Stokes equations (\ref{eq compressible NS}) for each velocity component are 
\begin{eqnarray}
 \partial_{t}u_i + u_j \partial_j u_i = \nu \partial_j\partial_j u_i -\partial_i p,
 \label{eq incompressible NS}\\
 \partial_{i}u_i=0,
 \label{eq free-divergence}
\end{eqnarray}
where we have taken $\rho=1$ and $f_i=0$ for simplicity. An acceptable solution for (\ref{eq incompressible NS})-(\ref{eq free-divergence}) are that functions $u_i,p\in\textsl{C}^{\infty}\left(\Omega\times\left[0,\infty\right)\right)$ such that have finite energy \cite{RefB3}, i.e.
\begin{equation}
 \int_{\Omega}u_i u_i\, d^3 x < \infty,\, \forall t\in\left[0,\infty\right).
\label{eq bounded energy}
\end{equation}
A variational method was developed by J. Leray  \textit{et al.} \cite{RefJ1} to solve (\ref{eq incompressible NS})-(\ref{eq free-divergence}). This method called \textit{weak solutions method} uses a free divergence test function that is multiplied by (\ref{eq incompressible NS}) with dot product. Then, an integration in $\textbf{x}\in\Omega$  and $t\in\left[0,T\right)$ is made. A theorem proved by J. Serrin (see \cite{RefJ2}) guarantees the existence of this smooth weak solutions in $t\in\left[0,T\right)$. In 1984, Beale, Kato and Majda (see \cite{RefJ3}) developed a theorem for 3-D Euler fluid (i.e. a fluid with $\nu=0$), which equation (\ref{eq incompressible NS}) has weak solutions in $t\in\left[0,T\right)$. This theorem states that if a initial smooth solution for Euler equation loses it smoothness some time later, then the maximum vorticity necesarily grows without bound as the critical time approaches. As immediate corollary of this theorem, the existence of a smooth solution in $t\in\left[0,T\right)$ is guaranteed. In 2007, Neustupa and Penel \cite{RefJ4}-\cite{RefJ5} prove that, taking boundary conditions
\begin{equation}
\vec{u}\cdot\hat{\vec{n}}=0,\,\,\vec{\nabla}\times\vec{u}\cdot\vec{\hat{n}}=0,\,\,\left[\vec{\nabla}\times\left(\vec{\nabla}\times\vec{u}\right)\right]\cdot\hat{\vec{n}}=0\,\,\,\textnormal{on}\,\, \partial\Omega,
\label{eq vorticity b-c}
\end{equation}
where $\hat{\vec{n}}$ is normal to surface, the velocity and the curl of vorticity are related   
\begin{equation}
      \int_{\Omega}\left|\vec{u}\right|^{2}\,d^3x<\infty \Leftrightarrow \left|\int^{t}_{0}\int_{\Omega}\left|\vec{\nabla}\times\left(\vec{\nabla}\times\vec{u}\right)\right|^2\,d^3 x\,d\tau\right|<\infty,
      \end{equation}
where $\Omega\subset\bf {R}^{3}$ and $t\in\left[0,\infty \right)$. The problem here, is that (\ref{eq vorticity b-c}) imposes a restriction on the space second derivative of the velocity, while (\ref{eq incompressible NS}) is second order in space.    

\section{Blow-up in 3-D Navier-Stokes Fluid.}
\label{sec:1}
In this section, we will prove a similar theorem to Beale-Kato-Majda theorem, but in this case there is a 3-D Navier-Stokes fluid instead of an Euler Fluid. In our case, we will prove that, taking suitable boundary conditions for vorticity and velocity, vorticity blows-up if, and only if, kinetic energy of a fluid particle blows-up $\forall t \in\left[0,\infty\right)$.
\begin{theorem}
Let be $u_i,p\in\textsl{C}^{\infty}\left(\Omega\times\left[0,\infty\right)\right)$ such that (\ref{eq incompressible NS})-(\ref{eq free-divergence}) are held. And let be the boundary conditions
\begin{equation}
\vec{u}\cdot\vec{\hat{n}}=0,\,\,\left(\vec{\nabla}\times\vec{u}\right)\cdot\left(\vec{u}\times\hat{\vec{n}}\right)=0,
\label{eq theorem}
\end{equation}
on $\partial\Omega$, where $\hat{\vec{n}}$ is normal to surface. Then,  
    \begin{equation}
      \int_{\Omega}\vec{u}\cdot\vec{u}\,d^3x<\infty \Leftrightarrow \left|\int^{t}_{0}\int_{\Omega}\left(\vec{\nabla}\times\vec{u}\right)\cdot\left(\vec{\nabla}\times\vec{u}\right)\,d^3 x\,d\tau\right|<\infty,
    \label{eq theorem 2} 
    \end{equation}
where $\Omega\subset\bf {R}^{3}$ and $t\in\left[0,\infty \right)$.      
\end{theorem}
\begin{proof}
We know that the pressure $p$ is the average kinetic energy of the molecules of the fluid that are colliding against the walls in a stagnant fluid. When the fluid is on stream, a new contribution due to kinetic energy of the flow takes part. So, we can add kinetic energy to pressure, subtracting
\begin{equation}
\partial_{i}\left(\frac{1}{2}\rho u_j u_j \right), \rho=1,
\label{eq proof 1}
\end{equation} 
in right and left side of (\ref{eq incompressible NS}). Then the equations (\ref{eq incompressible NS}) can be rewritten like
\begin{equation}
\partial_{t}u_i + u_j \partial_j u_i-u_j \partial_i u_j = \nu \partial_j\partial_j u_i -\partial_i \left(p+\frac{1}{2} u_j u_j\right).
\label{eq proof 2}
\end{equation}
Taking the scalar product of (\ref{eq proof 2}) with $\vec{u}$, we have
\begin{equation}
u_i \partial_{t}u_i + u_i u_j \partial_j u_i-u_i u_j \partial_i u_j = \nu u_i \partial_j\partial_j u_i -u_i \partial_i \left(p+\frac{1}{2} u_j u_j\right);
\label{eq proof 3}
\end{equation}
and taking the laplacian of both sides of (\ref{eq proof 2}), we have
\begin{equation}
\partial_{t}\partial_{i} u_i + \partial_i u_j\left( \partial_j u_i- \partial_i u_j\right)+u_j\partial_j\partial_i u_i-u_j\partial_i\partial_i u_j = \nu \partial_j\partial_j \partial_i u_i -\partial_i\partial_i \left(p+\frac{1}{2} u_j u_j\right).
\label{eq proof 4}
\end{equation}
Given that the fluid is free-divergence as (\ref{eq free-divergence}) says, the first and third terms in the left side of (\ref{eq proof 4}) and first term in right side of (\ref{eq proof 4}) vanishes, while the fourth term in left side of (\ref{eq proof 4}) is the first term in right side of (\ref{eq proof 3}). So substituting (\ref{eq proof 4}) in (\ref{eq proof 3}), we have
\begin{equation}
u_i \partial_{t} u_i = \nu \partial_i u_j\left(\partial_j u_i-\partial_i u_j\right)\\ +\nu \partial_i\partial_i \left(p+\frac{1}{2} u_j u_j\right)-u_i\partial_i \left(p+\frac{1}{2} u_j u_j\right).
\label{eq proof 5}
\end{equation}
But Levi-Civita tensor obeys $\epsilon_{ijk} \epsilon_{klm} =\delta_{il}\delta_{jm}-\delta_{im}\delta_{jl}$, so we can rewrite (\ref{eq proof 5}) like
\begin{equation}
u_i \partial_{t}u_i  = -\nu \epsilon_{ijk} \partial_i u_j \epsilon_{klm} \partial_l u_m+
\nu  \partial_i \partial_i \left(p+\frac{1}{2} u_j u_j\right) -\partial_i\left[u_i\left(p+\frac{1}{2} u_j u_j\right)\right].
\label{eq proof 6}
\end{equation}
Changing this equation to vectorial notation, it results: 
\begin{equation}
\partial_{t}\left(\frac{1}{2} \vec{u}\cdot \vec{u}\right) = -\nu \left(\vec{\nabla}\times\vec{u}\right)\cdot\left(\vec{\nabla}\times\vec{u}\right)+
\nu  \Delta \left(p+\frac{1}{2} \vec{u}\cdot \vec{u}\right) -\vec{\nabla}\cdot\left[ \vec{u}\left(p+\frac{1}{2} \vec{u}\cdot\vec{u}\right)\right].
\label{eq proof 7}
\end{equation}
If we integrate in the space region $\Omega$ and time interval $\left[0,t\right)$, then
\begin{eqnarray}
\int_{\Omega}\left(\frac{1}{2} \vec{u}\cdot \vec{u}\right)d^{3}x =
-\nu \int^{t}_{0}\int_{\Omega}\left(\vec{\nabla}\times\vec{u}\right)\cdot\left(\vec{\nabla}\times\vec{u}\right)d^3x d\tau \nonumber\\
+\nu\int^{t}_{0}\int_{\partial\Omega}\vec{\nabla}\left(p+\frac{1}{2} \vec{u}\cdot \vec{u}\right)\cdot\vec{\hat{n}}\,d^2xd\tau  -\int^{t}_{0}\int_{\partial\Omega}\left(p+\frac{1}{2} \vec{u}\cdot\vec{u}\right)\vec{u}\cdot\vec{\hat{n}}\,d^2xd\tau.\nonumber\\
\label{eq proof 8}
\end{eqnarray}
We have used the divergence theorem in the last two terms of right side of (\ref{eq proof 8}), where $\vec{\hat{n}}$ is the vector normal to the surface $\partial\Omega$. Choosing the Neumann boundary condition
\begin{equation}
\vec{\hat{n}}\cdot\vec{u}=0 \,, \forall\vec{x}\in\partial\Omega,
\label{eq proof 9}
\end{equation}
the last term of (\ref{eq proof 8}) vanishes. However, the second integral of right side of (\ref{eq proof 8}) vanishes if 
\begin{equation}
\vec{\hat{n}}\cdot\vec{\nabla}\left(p+\frac{1}{2} \vec{u}\cdot \vec{u}\right) =0 \,, \forall\vec{x}\in\partial\Omega.
\label{eq proof 10}
\end{equation}
We can obtain (\ref{eq proof 10}) with the next considerations.
Taking into account (\ref{eq proof 2}) and the identity
\begin{equation}
\Delta\vec{u}=\vec{\nabla}\left(\vec{\nabla}\cdot\vec{u}\right)-\vec{\nabla}\times\vec{\nabla}\times\vec{u},
\label{eq proof 11}
\end{equation}
we have
\begin{equation}
\partial_{t}\vec{u}-\vec{u}\times\left(\vec{\nabla}\times\vec{u}\right)=-\nu\vec{\nabla}\times\left(\vec{\nabla}\times\vec{u}\right)-\vec{\nabla}\left(p+\frac{1}{2} \vec{u}\cdot \vec{u}\right).
\label{eq proof 12}
\end{equation}
Since the surface $\partial\Omega$ is closed and don't change its shape with time (that is to say $\partial_{t}\vec{\hat{n}}=0$), integrating (\ref{eq proof 12}) in $\partial\Omega$ results
\begin{equation}
\int_{\partial\Omega}\left(\vec{\nabla}\times\vec{u}\right)\cdot\left(\vec{\hat{n}}\times\vec{u}\right)\,d^2x=\nu\int_{\partial\Omega}\vec{\hat{n}}\cdot\vec{\nabla}\left(p+\frac{1}{2} \vec{u}\cdot \vec{u}\right)\,d^2x,
\label{eq proof 13}
\end{equation}
where we have used
\begin{equation}
\partial_{t}\left(\vec{\hat{n}}\cdot\vec{u}\right)=0 \,, \forall\vec{x}\in\partial\Omega,
\label{eq proof 14}
\end{equation}
and the divergence theorem
\begin{equation}
\oint_{\partial\Omega}\left[\vec{\nabla}\times\left(\vec{\nabla}\times\vec{u}\right)\right]\cdot\vec{\hat{n}}\,d^2x=\int_{\Omega}\vec{\nabla}\cdot\left[\vec{\nabla}\times\left(\vec{\nabla}\times\vec{u}\right)\right]\,d^3x=0
\label{eq proof 15}
\end{equation}
So, (\ref{eq proof 10}) is equivalent to
\begin{equation}
\left(\vec{\nabla}\times\vec{u}\right)\cdot\left(\vec{\hat{n}}\times\vec{u}\right)=0 \,, \forall\vec{x}\in\partial\Omega.
\label{eq proof 16}
\end{equation}  
Taking the boundary condition (\ref{eq theorem}), the surface integrals of (\ref{eq proof 8}) vanishes. As a result,
\begin{equation}
\frac{1}{2}\int_{\Omega}\vec{u}\left(\vec{x},t\right)\cdot\vec{u}\left(\vec{x},t\right)d^3 x=-\nu \int^{t}_{0}\int_{\Omega}\left(\vec{\nabla}\times\vec{u}\left(\vec{x},\tau\right)\right)\cdot\left(\vec{\nabla}\times\vec{u}\left(\vec{x},\tau\right)\right)d^3 x d\tau,
\label{eq proof 17}
\end{equation}
and the proof of the theorem follows \qed.
\end{proof}
Three consequences arise here choosing boundary coditions (\ref{eq theorem}) for velocity and vorticity. First of all, relation (\ref{eq proof 17}) means that kinetic energy in the volume diminishes with time, if vorticity is not null. This is perfectly reasonable since kinematic viscosity produces heat from kinetic energy. Second, if the fluid is an Euler fluid, i.e. $\nu=0$, the fluid don't disipate kinetic energy and the kinetic energy don't depends explicitly on time. And third, vorticity controls the blow-up in time.\\
Boundary coditions (\ref{eq theorem}) are 
less restrictive than (\ref{eq vorticity b-c}), because the firsts are space first derivative and the seconds are higher order derivatives.
The theorem proposed is a similar result to Beale-Kato-Majda theorem, with Navier-Stokes fluid instead of Euler fluid. But there is not a proof that (\ref{eq theorem 2}) be a more general result of Beale-Kato-Majda theorem in $t\in\left[0,\infty\right)$, because right hand side of (\ref{eq proof 10}) vanishes in Euler fluid and because (\ref{eq theorem}) condition has not sense in Euler fluids. This theorem is more powerful than the Leray weak solutions method in the temporal domain, but more weak in the space domain due to boundary conditions.

%
%
%

\end{document}